\documentclass[]{article}

\usepackage{graphicx,subfig,tikz,amsthm,amssymb,amsmath,amsfonts,mathtools}

\newtheorem{lemma}{Lemma}
\newtheorem{theorem}{Theorem}
\newtheorem{remark}{Remark}

\title{Self-Coordinated Corona Graphs: a model for complex networks} %%%for recto running head
\author{Rohan Sharma, Bibhas Adhikari\\Centre for System Science, Indian Institute of Technology Jodhpur, India\\Department of Mathematics, Indian Institute of Technology Kharagpur, India} %%% for verso running head
\date{}
\begin{document}
\maketitle
\begin{abstract}
{Recently, real world networks having constant/shrinking diameter along with power-law degree distribution are observed and investigated in literature. Taking an inspiration from these findings, we propose a deterministic complex network model, which we call Self-Coordinated Corona Graphs (SCCG), based on the corona product of graphs. As it has also been established that self coordination/organization of nodes gives rise to emergence of power law in degree distributions of several real networks, the networks in the proposed model are generated by the virtue of self coordination of nodes in corona graphs. Alike real networks, the SCCG inherit motifs which act as the seed graphs for the generation of SCCG.  We also analytically prove that the power law exponent of SCCG is approximately $2$ and the diameter of SCCG produced by a class of motifs is constant. Finally, we compare different properties of the proposed model with that of the BA and Pseudofractal scale-free models for complex networks.}
%{Corona graphs; degree distribution; diameter; clustering coefficient; motifs; disassortativity.}
\end{abstract}
%\keywords{Corona graphs; degree distribution; diameter; clustering coefficient; motifs; disassortativity.}

\section{Introduction}
Networks are everywhere around us. Due to their inherent non-trivial topological and functional characteristics, some of them are called complex networks, for instance, biological networks, food webs, internet, the world wide web, financial markets and social networks \cite{albert2002statistical,boccaletti2006complex,newman2003structure,dorogovtsev2002evolution,estrada2011structure}. The study of the facts responsible for structural and dynamical behaviour of these networks has generated a lot of interest in recent past. Several deterministic and stochastic models are proposed to replicate the structure and dynamics of real world complex networks supported with empirical data \cite{dorogovtsev2002pseudofractal,ravasz2003hierarchical,ravasz2002hierarchical,barabasi2001deterministic,barabasi1999emergence,ispolatov2005duplication,leskovec2010kronecker,dorogovtsev2000evolution,newman2002assortative,jung2002geometric}. The comparison of these models with real networks is done usually by comparing different properties of the networks including degree distribution, diameter or average path length, average clustering coefficient, frequencies of triangles, degree-degree correlation etc.

The Erd{\H o}s R{\'e}nyi model (henceforth denoted by `ER') \cite{boccaletti2006complex,erdds1959random} is regarded as one of the earliest random network models proposed in the literature. In this model, a fixed set of nodes $N$ is chosen as input and then $M$ edges are carved out randomly from these nodes while restricting multi-edges. However, it lacks the features of the real world networks, for example, power law degree distribution. A preliminary model exhibiting the power law degree distribution was proposed by Barab{\'a}si et al., known as the Barab{\'a}si-Albert (henceforth denoted by `BA') model \cite{barabasi1999emergence}. This is a stochastic model generated by two steps: (1) growth of the network with the addition of new nodes (2) probabilistic preferential attachment of the new nodes to the existing nodes based on their degree. This model lacks the degree-degree correlation of the nodes \cite{gross2009adaptive} which is a genuine property of most of the real networks, and frequencies of triangles(responsible for network transitivity) in the model is very less \cite{durak2012degree} which contradicts the property of most real networks. In \cite{barabasi2001deterministic}, the authors remarked that stochastic models lack any visual understanding of the scale-freeness and hence they introduced a hierarchical deterministic scale-free model. However, this model lacks clustering property of real-world networks. Another hierarchical deterministic scale-free model was presented in \cite{ravasz2003hierarchical} with high clustering coefficient, however other important properties are unexplored. Pseudo-fractal scale free model(henceforth denoted by `PFSF') is an another deterministic scale free model introduced in \cite{dorogovtsev2002pseudofractal}, in which, the network starts from a single edge at time $t=-1$ and at each timestep $t,$ ($\forall$ $t\ge 0$) new nodes are added to the network corresponding to each edge of the network generated at time $t-1$. All the above mentioned models except ER model has increasing diameter with the addition of the new nodes to the existing network which is being contradicted by an observation of constant/shrinking diameter phenomenon of certain real networks reported in \cite{leskovec2007graph}. In \cite{leskovec2010kronecker}, the authors have presented a deterministic and a stochastic model having the constant and shrinking diameter property respectively based on Kronecker product of graphs. Such networks are called Kronecker graphs (henceforth denoted by `KG').

\begin{table}
\centering
  \captionsetup{justification=raggedright,singlelinecheck=false}
  \caption{A statistical snapshot of the few published real world networks having shrinking diameter as investigated in \cite{leskovec2007graph} where $N$, $E$, $\overline{C}$, $r$  and $\gamma$ represents number of nodes, edges, average clustering coefficient, assortativity coefficient and power law exponent respectively. Collaboration network (AstroPh, cond-mat), citation network (cit-HepPh, cit-HepTh), autonomous systems(caida20071105) and US Patents are the used datasets. Blank entries indicate the data is unavailable. }
  \label{tbl:datasets}
  \begin{tabular}{l l l l l l l}
    \hline\hline
    Datasets & $N$ & $E$ & $\overline{C}$ & $r$ & $\gamma$ & Ref(s).\\ \hline
    AstroPh & $18K$ & $198K$ & $0.63$ & $0.2$ & $2.86$ & \cite{durak2012degree,konect:2015:ca-AstroPh}\\ 
    cond-mat & $38K$ & $58K$ & & $-0.12$ & $2.8$ & \cite{konect:2015:opsahl-collaboration,konect:newman01}\\
    HepPh & $34K$ & $420K$ & $0.3$ & $-0.006$ & $3.04$ & \cite{durak2012degree,konect:2015:cit-HepPh}\\
    HepTh & $27K$ & $352K$ &  & $-0.03$ & $3.12$ & \cite{konect:2015:cit-HepTh}\\
    caida20071105 & $26K$ & $53K$ & $0.21$ & $-0.19$ & $2.09$& \cite{durak2012degree,konect:2015:as-caida20071105}\\
    US Patents & $3,774K$ & $16,518K$ & & $0.16$ & $4$ & \cite{b376,konect:2015:patentcite}\\
    \hline      
  \end{tabular}%  \captionsetup{justification=raggedright,singlelinecheck=false}
\end{table}

In this paper, we propose a deterministic complex network model by gainfully using the concept of self co-ordination of nodes in a corona graph which is proposed in \cite{sharma2015spectra,sharma2015corona}. The corona graphs are modelled by using corona product of two graphs \cite{parsonage2011generalized,frucht1970corona,sharma2015spectra,sharma2015corona}. It is proved in \cite{sharma2015corona} that the degree distribution of corona graphs do not follow power law, however, the degree distribution decays exponentially when the seed graph is regular. Also, the diameter of a corona graph increases in each iteration of the formation of the corona graph for any seed graph. As it is observed that self coordination among the nodes help emerge the scale-free degree distribution in real world networks \cite{barabasi1999emergence}, we show that self coordination of nodes in corona graphs produces networks with hubs having a power-law degree distribution. We assume that the nodes which get linked to the network after each corona product, self-coordinate among themselves and select hubs from the nodes added in the preceding step of the formation of the corona graph. To the best of our knowledge, this is the first model of its kind to model complex networks which uses self-coordination as the prime concept for emergence of  power law degree distribution. We call the complex networks generated by this method as self-coordinated corona graphs (`SCCG'). We show that the power law exponent of degree distribution in SCCG is approximately $2$. It is interesting to observe that SCCG inherit motifs, for instance, the seed graph, which defines the SCCG, is a motif. We further show that SCCG have other significant properties of real networks including high clustering, constant diameter, disassortative behaviour, and high frequency of triangles. We mention that several real world networks, for example, the networks mentioned in Table ~\ref{tbl:datasets} are having power law degree distribution with shrinking diameter \cite{durak2012degree, leskovec2007graph} phenomena. In addition, the assortativity and disassortativity mixing pattern have  also been observed in these networks. Thus, SCCG can be a potential model to capture various structural and dynamical properties of these networks. We have also compared the above mentioned properties of the SCCG with the properties of networks generated by BA and PFSF models.

We organize the paper as follows. In Sec.~\ref{sec:CG}, we describe the corona graphs proposed in \cite{sharma2015spectra}. Sec.~\ref{sec:SCCG} and Sec.~\ref{sec:prop} include the algorithm for generation of SCCG, proof of the degree sequence of the hubs in SCCG, an analytic proof of the power law exponent of the hubs, and the constant diameter phenomena. In Sec.~\ref{sec:con}, we conclude the content of the paper.

%%%%%%%%%%%%%%%%%%%%%%%%%%%%%% % % % % % % % % % % % % % % % % % % % % % % % % % % % % %

\section{\label{sec:CG}Corona graphs}
In this section, we discuss the generation of corona graphs and their properties. It should be noted that the corona graphs and the subsequent improvement of this model in the next section are only applicable for simple graphs.
 
Let $G=(V,E)$ be a graph with the node set $V=\{v_1,...,v_n\}$ such that $|V|=n$, and $E$ the set of edges. The adjacency matrix $A(G)=[a_{v_i v_j}]$ of dimension $|V|\times |V|$ associated with $G$ is defined by $a_{v_i v_j}=1$ if $(v_i,v_j)\in E$, and $a_{v_i,v_j}=0$ otherwise. The columns and rows of $A(G)$ are labelled by the nodes of $G$. 

Let $G^{(0)}= G$ be the fixed seed graph for the corona graphs \cite{sharma2015spectra,sharma2015corona} to be generated by defining
\begin{equation}
G^{(m+1)}=G^{(m)}\circ G
\end{equation}
where $m(\ge 0)$ is a large natural number and $\circ$ denotes the corona product of graphs.

The properties of the corona graphs are as follows
\begin{enumerate}
\item[1.] The number of nodes in $G^{(m)}$ for any positive integer $m$ is \cite{sharma2015spectra,sharma2015corona}
\begin{equation}
|V^{(m)}|=n(n+1)^{m}.
\end{equation}
\item[2.] The number of nodes added in $G^{(i)} = G^{(i-1)}\circ G$ is $n^2 (n+1)^{i-1}$ where $i\ge 1$.\cite{sharma2015corona}
\item[3.] If $|E|$ and $n$ are the number of edges and nodes in the seed graph $G^{(0)}$ respectively, then the number of edges in $G^{(m)}$ is \cite{sharma2015spectra,sharma2015corona}
\begin{equation}\label{eqn:3}
|E^{(m)}|=(|E|+(|E|+n)((n+1)^m-1)).
\end{equation}
\item[4.] It is evident from the \cite{sharma2015corona} that the diameter of a corona graph $(G^{(m)})$ is $D^{(0)}+2m$ where $D^{(0)}$ is the diameter of the corresponding seed graph $G^{(0)}$. 
\end{enumerate}

Although the corona graph model has fat tailed degree distribution, which is found abundantly in many real world networks, it lacks the power-law degree distribution which is a non-trivial characteristic associated with most of the real world networks \cite{barabasi1999emergence, sharma2015spectra, sharma2015corona, misiewicz2011fat, rachev2003handbook}. Moreover, as it follows from above that the diameter of corona graphs is monotonically increasing, it contradicts the recent observation of constant/shrinking diameter in real world networks \cite{leskovec2007graph}. These drawbacks in the model of corona graphs motivated us to introduce the idea of self-coordination in corona graphs.

% % % % % % % % % % % % % % % % % % % % % % % % % % % % %

\section{\label{sec:SCCG}Self-coordinated corona graphs}
In this section, we first discuss the concept of self-coordination and its significance in real world networks. We elaborate this concept in the context of self-coordination of nodes in the corona graphs. The self-coordination refers with the links made by the new nodes added in  a network with the existing nodes, especially hubs, in such a manner that leads to the emergence of empirically observed properties of real world networks, for example, in power law degree distribution. This power-law behaviour is observed, for instance, in collaboration networks, citation network, autonomous systems etc. as mentioned in Table ~\ref{tbl:datasets}.

Let the self-coordination after every $G^{(i)}=G^{(i-1)}\circ G$ (where $i\ge 1$) be denoted by $SC^{(i)}$. By taking an inspiration of the self-coordination described above, we propose the algorithm such that the nodes of the corona graphs self-organize among themselves after each $G^{(i)}$ ($\forall$ $i\in[1,m]$). The self-organization happens between $n^2(n+1)^{i-1}$ nodes added in the $i^{th}$ step during formation of $G^{(m)}$ $\forall$ $i\le m$ of $G^{(i)}$ and the selected hub nodes which are chosen based on their popularity in the existing network. We consider popularity as the degree of the node. If all nodes have same degree then any node could be designated as a hub.

\begin{figure}
\centering
  \includegraphics[width=0.5\textwidth]{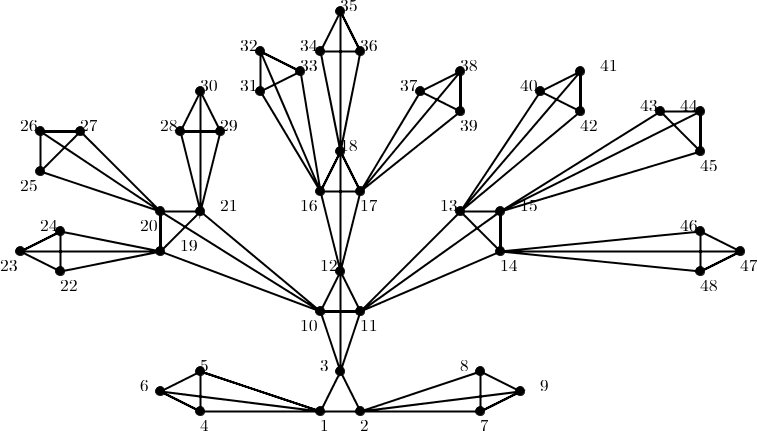}
\captionsetup{justification=raggedright,singlelinecheck=false}
\caption{\label{3}A small section of well-labelled corona graph generated with seed graph $G^{(0)}=K_3$.}
\end{figure}

We discuss few preliminaries that would help in comprehending the terminologies used in the algorithm, for instance, the concept of immediate ancestor hubs. Let the seed graph be $G^{(0)}=K_3$ for generation of the corona graphs. Fig.~\ref{3} represents a small section or part of these corona graphs. Let nodes $3,12,15,18,20$ be the hubs to which new nodes attach during the self-organization step. Nodes $4,\ldots,12$ have the immediate ancestor hub as $3$ while the nodes $13,\ldots,21$ have $3,12$ as ancestor hubs. Similarly, nodes $22,\ldots,30$ have  $3,12,20$; $31,\ldots,39$ have $3,12,18$; and $40,\ldots,48$ have $3,12,15$ as ancestor hubs respectively.

The second concept that we use in the algorithm is the set of skeletal(SK) and the offshoot(OS) nodes. The nodes of $G^{(0)}$ in Step $0$ of the algorithm are considered as skeletal nodes. The nodes which get attached first time to a skeletal node in a $G^{(i)}$ (where $i\ge 1$) will be considered as skeletal nodes whereas the new nodes which get attached to the skeletal node second time onwards will be considered as offshoot nodes. The nodes which get linked with an offshoot node will always be the considered as offshoot nodes. We elaborate this in Fig.~\ref{4}. We denote the set of skeletal nodes by SK and off-shoot nodes by OS respectively.

\begin{figure}
\centering
\subfloat[\label{4a}]{
  \includegraphics[width=0.05\textwidth]{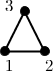}
}
\hspace{0mm}
\subfloat[\label{4b}]{
  \includegraphics[width=0.06\textwidth]{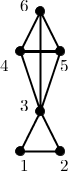}
}
\hspace{0mm}
\subfloat[\label{4c}]{
  \includegraphics[width=0.3\textwidth]{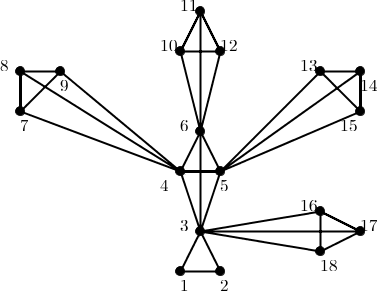}
}
\hspace{0mm}
\subfloat[\label{4d}]{
  \includegraphics[width=0.45\textwidth]{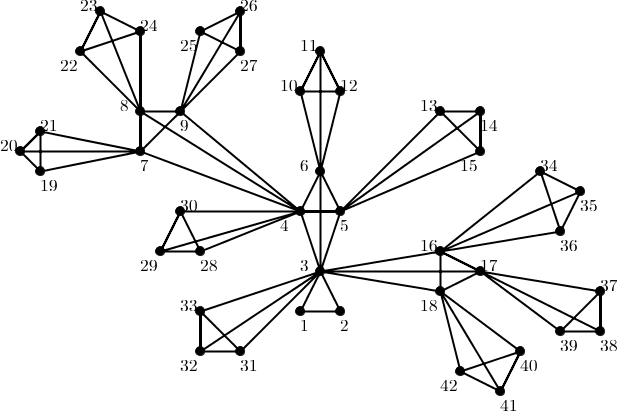}
}
\captionsetup{justification=raggedright,singlelinecheck=false}
\caption{\label{4}A small section of corona graphs (a) $G^{(0)}=K_3$ (b) $G^{(1)}$ (c) $G^{(2)}$ (d) $G^{(3)}$.}
\end{figure}
\begin{figure}
\subfloat[\label{5a}]{
  \includegraphics[width=0.05\textwidth]{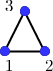}
}
\hspace{5mm}
\subfloat[\label{5b}]{
  \includegraphics[width=0.4\textwidth]{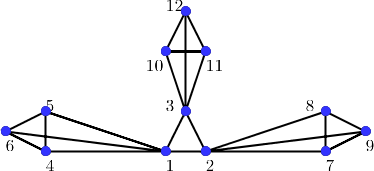}
}
\hspace{5mm}
\subfloat[\label{5c}]{
  \includegraphics[width=0.4\textwidth]{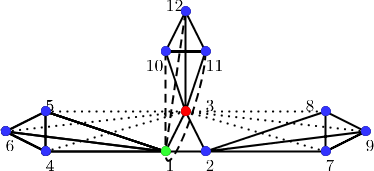}
}
\hspace{0mm}
\subfloat[\label{5d}]{
  \includegraphics[width=0.4\textwidth]{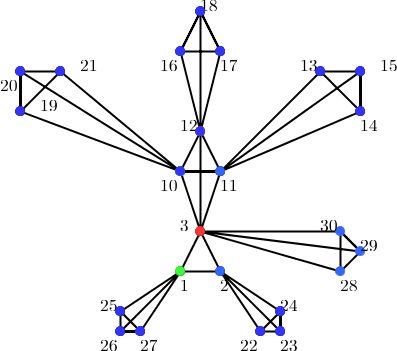}
}
\hspace{20mm}
\subfloat[\label{5e}]{
  \includegraphics[width=0.4\textwidth]{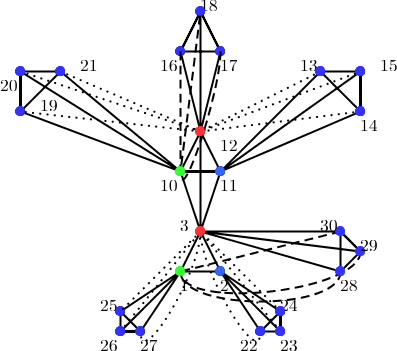}
}
\captionsetup{justification=raggedright,singlelinecheck=false}
\caption{\label{5} (a) In $G^{(0)}=K_3$, all are simple nodes. (b) In $G^{(1)}$, all nodes are simple nodes. (c) In $SC^{(1)}$, node $3$ is hub, node $1$ is non-hub and rest are simple nodes. (d) This sub-figure is showing a small part of $G^{(2)}$ with nodes $\{4,5,\ldots,9\}$ and the corresponding edges to hubs and non-hubs are removed to avoid complexity. (e) Node set $A=\{13,14,15,19,20,21\}$ and $B=\{16,17,18\}$ select node 12 as hub and node 10 as non-hub. $A\bigcup B$ would also connect with immediate ancestor hub $3$. Since, node $\{3\}$ and $\{1\}$ are already decided as hub and non-hub by the node set $C=\{4,5,\ldots,12\}$ and hence, $\{22,23,\ldots,30\}$ are connected to them accordingly. The red, green and blue colour nodes in the figure are hubs, non-hubs and simple nodes respectively.}
\end{figure}

Suppose an arbitrary skeletal node is denoted by $n_x$ and initially $SK=OS=\emptyset.$ Consider the seed graph $G^{(0)}=K_3$. Initially, in $G^{(0)}$, nodes $SK=SK$ $\bigcup$ $\{1,2,3\}$ and $OS=\emptyset$. In $G^{(1)}$, $(3,n_x)=0$, hence, $SK=SK$ $\bigcup$ $\{4,5,6\}$ and $OS=\emptyset$. Similarly, in $G^{(2)}$, $(4,n_x)=(5,n_x)=(6,n_x)=0$, hence $SK=SK$ $\bigcup$ $\{7,\ldots,15\}$. But $(3,n_x)=1$, hence $OS=OS$ $\bigcup$ $\{16,17,18\}$. Similarly, in $G^{(3)}$, $(7,n_x)=(8,n_x)=(9,n_x)=0$, $(4,n_x)=1$ and $(3,n_x)=1$, hence $SK=SK$ $\bigcup$ $\{19,\ldots,27\}$, $OS=OS$ $\bigcup$ $\{28,29,30\}$ and $OS=OS$ $\bigcup$ $\{31,32,33\}$. Now, nodes $A=\{34,\ldots,42\}$ are connected to the offshoot nodes, hence, $OS=OS$ $\bigcup$ $A$. The enrichment of sets $SK$ and $OS$ would be continued in the same way for future $G^{(i)}$.

The third concept is that there are $3$ types of nodes in SCCG -- $(1)$ hubs $(2)$ non-hubs and $(3)$ simple nodes. We would elaborate about them with the help of Fig.~\ref{5} which is again generated with $G^{(0)}=K_3$ as the seed graph. Initially, in $G^{(0)}$, all nodes are considered as simple nodes (shown in blue colour in Fig.~\ref{5a}). In $G^{(1)}$, all nodes are again considered as simple nodes (as in Fig.~\ref{5b}). Now, during $SC^{(1)}$, node $3$ is selected as hub (shown in red colour in Fig.~\ref{5c}) and hence, nodes $B=\{4,5,\ldots,9\}$ are connected to it. Since, nodes $C=\{10,11,12\}$ are already connected to the hub, so it will get linked with any other node choosing from $\{1,2\}$. This choice could be done again based on degree or randomly or deterministically. Since we are generating SCCG deterministically, we have chosen node $\{1\}$ deterministically and hence, elements of $C$ would connect with $\{1\}$ and this node is to be designated as non-hub (shown in green colour in Fig.~\ref{5c}) while node $\{2\}$ and $\{4,5,\ldots,12\}$ would remain as simple nodes (as in Fig.~\ref{5c}). The same methodology would be followed for any $G^{(i)}$ followed by $SC^{(i)}$. This hub or non-hub selection process is done only by skeletal nodes and those offshoot nodes which are connected directly with the skeletal nodes (shown in Fig.~\ref{5d} and Fig.~\ref{5e} and the explanation is provided under the figure). For instance, in $G^{(i)}$ (where $i\ge 3$), the offshoot nodes attached to nodes $ \{10\}$, $\{11\}$ and $\{12\}$ would be connected to the their skeletal hubs $\{3,12\}$. Those offshoot nodes which are connected to offshoot nodes only would be connected to their immediate ancestor hubs of the skeletal nodes. For instance, in $G^{(i)}$ (where $i\ge 3$), any node attached to offshoots nodes $\{22,23,24\}$, $\{25,26,27\}$ and $\{28,29,30\}$ (shown in Fig.~\ref{5e}) would be connected to $\{2\}$, $\{1\}$ and $\{3\}$ respectively which are their immediate ancestor hubs in the skeletal. Hence, all the offshoot nodes are always simple nodes. 

Let $V^{(i)}$ denote the node set associated with the corona graph $G^{(i)}=G^{(i-1)}\circ G$ such that $|V^{(i)}|=n^2(n+1)^{i-1}$. Suppose the set of all nodes linked to skeletal nodes in $G^{(i)}$ is denoted by $X$ while the set of all offshoot nodes connected to skeletal nodes directly (i.e. adjacently) in $G^{(i)}$ is denoted by $Y$. We now provide our algorithm comprised of the above concepts as follows.

\subsection*{Algorithm}
\begin{enumerate}
\item[(1)] Perform $G^{(1)}=G^{(0)}\circ G^{(0)}$. In $SC^{(1)}$, the new  $|V^{(2)}-V^{(1)}|=n^2$ nodes self coordinate by selecting the hubs and non-hubs among the old nodes and get linked with them.
\item[(2)] Perform $G^{(i)}=G^{(i-1)}\circ G^{(0)}$ (where $i\ge 2$). In $SC^{(i)}$, nodes in $G^{(i)}$ self coordinate by the steps below \begin{enumerate} 
\item nodes in the skeletal node set $X$ and off-shoot node set $Y$ that appeared in the $i$-th step of the corona product select the hubs and non-hubs from the nodes appeared in the step $(i-1)$ (of the corona product) and get linked with them. 
\item The nodes in $X$ and $Y$ also get linked with their immediate ancestors hubs of the skeletal nodes appeared in step $i-1$ (of the corona product). The remaining offshoot nodes (i.e. $n^2(n+1)^{i-1}-(|X|+|Y|)$) which appeared in the $i$-th step of the corona product (but not directly linked with skeletal nodes of $i-1$-th step), get linked with their immediate ancestor hubs of the skeletal nodes in $G^{(i-1)}$.\end{enumerate}
\end{enumerate}

We mention that the idea of self coordination step as described in the algorithm above is inspired by the phenomena of existence of motifs in real world networks. Recall that, a motif is a small graph which is present frequently in the entire network and motifs are considered as the building blocks of a complex network \cite{yaveroglu2014ergm,milo2002network}. Motifs are identified in different real networks including transcription networks, worldwide web networks, and social networks as mentioned in \cite{yaveroglu2014ergm} and the references therein. In \cite{prvzulj2004modeling, prvzulj2007biological}, the authors generalized the concept of network motifs by referring them as graphlets and mentioned different types of graphlets having number of nodes $|V|\in \{2,3,4,5\}$ \cite{milo2002network, prvzulj2004modeling, prvzulj2007biological}. In the proposed model, a motif is the seed graph which defines the corona graphs. Hence, if one starts constructing SCCG by considering the seed graph as a clique($K_n$, $\forall n\ge 2$), tree($S_n$, $\forall n\ge 3$), circuit($C_n$, $\forall n\ge 4$), path($P_n$, $\forall n\ge 4$) or any arbitrary graph, the SCCG would inherit the chosen seed graph as a motif.
\section{\label{sec:prop}Properties of the SCCG}
In this section, we prove that the hubs of SCCG follow power law degree distribution with exponent approximately $2.$ We will also present here the diameter of SCCG (along their analytical proof) generated by a few special types of seed graphs.
\subsection{Degree Distribution}
The probability distribution of the degree of the nodes in a graph is considered as the degree distribution of the graph. In the proposed model, we describe the degree distribution in accordance with the degree of the hubs because the probability distribution of the graph largely depends on the degree of the hub nodes. First, we determine the degree sequence of the hubs of SCCG. 

Consider $G^{(m)}$ constructed by the seed graph $G^{(0)}$ of order $n.$ Suppose $k^0_{max}$ is the maximum degree of a node in $G^{(0)}.$ Then, a node $h$ in $SC^{(m)}$ is considered as a hub if the degree of $h$ lies in the interval $[k^0_{max}+m+n^2,k^0_{max}+n((n+1)^m-1)].$ %Then, we have the following result.

Note that there is only one hub in $SC^{(m)}$ of degree $k^0_{max}+n((n+1)^m-1)$ and this node is appeared in $G^{(0)}.$ Further, the increment in the degree of hub nodes for any hub in $SC^{(i)}$ is $\sum\limits_{l=1}^{i}n^2(n+1)^{l-1}$. We will use these facts while determining the degree of hubs as follows.
\begin{itemize}
\item[(a)] Hubs of degree $k^0_{max}+n((n+1)^m-1):$

This node is selected as the only hub in $SC^{(1)}$. In each $SC^{(i)}$, $i\ge 1$ the degree of this hub would be increased by $n^2(n+1)^{i-1}$. Hence, the degree of this hub in $SC^{(m)}$= (Degree of node in $G^{(0)})+($Edges from $V^{(1)})+($Edges from $V^{(2)})+\ldots+($Edges from $V^{(m)})=(k^0_{max})+(n^2)+(n^2(n+1))+\ldots+n^2(n+1)^{m-1}=k^0_{max}+n((n+1)^m-1)$.

\item[(b)] Hubs having degree $k^0_{max}+(m-j+1)+n((n+1)^j-1)$ where $1\le j\le m-1:$

Let $v$ be a hub appeared in the $(m-j)^{th}$ step of the network formation. Since a hub should be the node of maximum degree of the seed graph, the number of edges connected to a hub node is given by
\begin{eqnarray*} && (k^0_{max}+1)+(\mbox{Number of hubs up to}\,\, (m-j)^{th} \mbox{step}) +(\mbox{Number of links from the nodes}\\ && \mbox{appeared added after}\,\, \mbox{from} \,\, (m-j+1)^{th}\,\, \mbox{to} \,\, m^{th})\\
&=& k^0_{max}+1+(m-j)+\sum\limits_{l=1}^{m-j}n^2(n+1)^{l-1}+ \sum\limits_{l=1}^{j}n^2(n+1)^{l-1}\\
&=& k^0_{max}+(m-j+1)+n((n+1)^j-1).\end{eqnarray*}
\end{itemize}

It may be noted that for any hub $h\in V^{(0)}$ of SCCG, all the nodes added in the corona product process $G^{(i)}, i\ge 2,$ are adjacent to $h.$ Thus, for SCCG with seed graph $G^{(0)}=\{K_n,S_n\},$ where $n$ is the number of nodes in $G^{(0)}$, any hub in $SC^{(1)}$ is adjacent to all the nodes of the SCCG. In general, if $G^{(0)}$ is any seed graph having a node $h$ of degree $(n-1),$ $h$ is adjacent to all the nodes in the resultant SCCG.

It is also easy to verify that the number of skeletal nodes, which appear in the $i$-th step of the corona graph during the process of generating SCCG and get linked with skeletal nodes which appeared in the $i-1$-th step, is $n^{i+1}.$ Consequently, the frequency of hubs which appear in the $i$-th step of the formation of SCCG is $n^{i-1}, i\ge 1$ where $n$ represents the nodes in the seed graph $G^{(0)}.$

It follows from the above that there is a large difference between the degrees of the hubs. Thus, the power-law exponent could be calculated by using the cumulative degree distribution which can be derived from the cumulative frequency. If $k$ and $k^{'}$ are the degree instances of any two hubs of a resultant SCCG such that $k^{'}\ge k$, then the cumulative frequency of the hubs is given by
$$\nu(k^{'},k)= \frac{n(n^{m-j}-1)}{n-1}$$ where $j\in [1,m-1]$. %Hence we have the following theorem.

Let $j$ be a hub of degree $d_j$. % as given in Theorem \ref{Thm:ds}. 
Probability that a node having degree $k=d_j$ is\\
\begin{equation*} \label{eq1}
\begin{split}
(k=d_j) &= \dfrac{\dfrac{n(n^{m-j}-1)}{(n-1)}}{n(n+1)^m} \\
 & = \frac{(n^{m-j}-1)}{(n-1)(n+1)^m}
\end{split}
\end{equation*}

\noindent Now, as we had determined above the degree of hubs, %shown in Theorem \ref{Thm:ds}, 
$$d_j=k^0_{max}+(m-j+1)+n((n+1)^j-1).$$
Claim: 
\begin{equation*}
(k^0_{max}+(m-j+1)+n((n+1)^j-1))^{1-\gamma} = \frac{(n^{m-j}-1)}{(n-1)(n+1)^m}
\end{equation*}
If $m\gg 1$ and $j\rightarrow (m-1)$, $\dfrac{(n^{m-j}-1)}{(n-1)(n+1)^m}\approx (n+1)^{-m}$ and $n(n+1)^j\gg (k^0_{max}+(m-j+1)-n)$. Hence,
\begin{equation*}
(n(n+1)^j)^{1-\gamma} \approx (n+1)^{-m}.
\end{equation*}

\noindent Taking logarithm in both sides, we get
\begin{equation*}
(1-\gamma)=\frac{-m\ln(n+1)}{\ln n+j\ln(n+1)}
\end{equation*}
Since $m\gg 1$ and $j\rightarrow (m-1)$, hence $j\gg 1$, $j\ln(n+1)\gg \ln n$ and $\frac{m}{j}\rightarrow 1$. Therefore, 
\begin{equation*}
\begin{split}
(1-\gamma) &\approx \frac{-m\ln (n+1)}{j\ln(n+1)}\\
		\gamma&\approx 2.
\end{split}
\end{equation*}

\subsection{Density of SCCG}
Density of a network is an important concept to determine the sparsity of the network. The density of an undirected network $G=(V,E)$ is defined by \cite{laurienti2011universal}
\begin{equation}
d=\frac{2|E|}{|V|(|V|-1)}
\end{equation}
where $|E|$ and $|V|$ represent the number of edges and nodes respectively in $G$. It is evident that $0<d\le 1.$ If $d\ll 1$, then the network is sparse. It is noted in \cite{i2003optimization} by the authors that many real world network have $d\in[10^{-5},10^{-1}]$. In \cite{zhou2006hierarchical}, the authors remarked by specifying the example of neural networks that sparsity saves the energy of the network without effecting its functionality. The network sparsity is also observed in other biological networks like metabolic network (\cite{wagner2001small} and the references therein). 

Here, we numerically calculate the densities of different SCCG generated by a few seed graphs (Fig.~\ref{6}) having $12,000$ nodes in Table \ref{tbl:density}. The Table shows that sparsity of SCCG is similar to that of real networks.
\begin{figure}
\centering
\subfloat[\label{6a}$G_1$]{
  \includegraphics[width=0.05\textwidth]{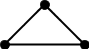}
}
\hspace{2mm}
\subfloat[\label{6b}$G_2$]{
  \includegraphics[width=0.05\textwidth]{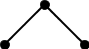}
}
\hspace{2mm}
\subfloat[\label{6c}$G_3$]{
  \includegraphics[width=0.04\textwidth]{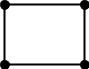}
}
\hspace{2mm}
\subfloat[\label{6d}$G_4$]{
  \includegraphics[width=0.04\textwidth]{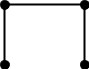}
}
\hspace{2mm}
\subfloat[\label{6e}$G_5$]{
  \includegraphics[width=0.045\textwidth]{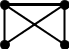}
}
\hspace{2mm}
\subfloat[\label{6f}$G_6$]{
  \includegraphics[width=0.045\textwidth]{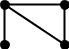}
}
\captionsetup{justification=raggedright,singlelinecheck=false}
\caption{\label{6} Seed graphs used for determining the density in Table ~\ref{tbl:density}.}
\end{figure}
\begin{table}[h]
	\captionsetup{justification=raggedright,singlelinecheck=false}
	\caption{Density(d) of different SCCGs (having $12,000$ nodes) generated by seed graphs of Fig.~\ref{6}.}
	\label{tbl:density}
		\begin{tabular}{l r | l r | l r}
		\hline\hline
		$G^{(0)}$ & Density(d) & $G^{(0)}$ & Density(d) & $G^{(0)}$ & Density(d)\\ \hline
		{$G_1$} & $8.46\times 10^{-4}$ & {$G_2$} & $7.92\times 10^{-4}$ & {$G_3$} & $8.4\times 10^{-4}$\\
		{$G_4$} & $8\times 10^{-4}$ & {$G_5$} & $8.8\times 10^{-4}$ & {$G_6$} & $8.4\times 10^{-4}$\\
		\hline
		\end{tabular}
\end{table}
\begin{figure}
\centering
\subfloat[\label{8a}]{
  \includegraphics[width=0.35\textwidth]{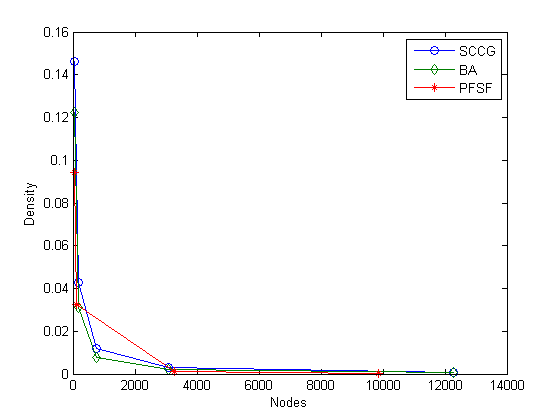}
}
\hspace{0mm}
\subfloat[\label{8b}]{
  \includegraphics[width=0.35\textwidth]{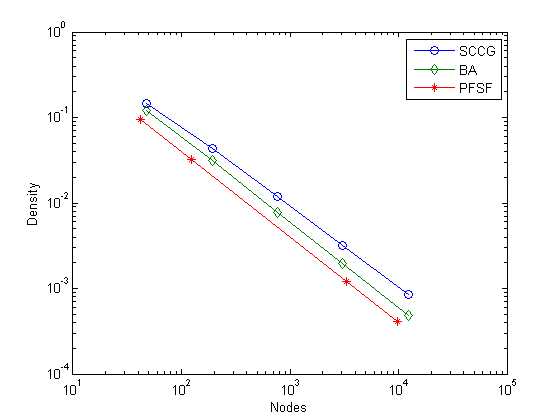}
}
\captionsetup{justification=raggedright,singlelinecheck=false}
\caption{\label{8} Comparison of the density of the SCCG, BA and FSF model (a) Linear plot (b) Loglog plot.}
\end{figure}
The Fig.\ref{8} is displaying the density of the SCCG, BA and FSF model with $G^{(0)}=K_3$ as the seed graph. 
\subsection{Diameter}
Diameter of a network is the longest shortest path between any pair of nodes in the network. As usual, finding the exact formula of the diameter of SCCG generated by any seed graph $G^{(0)}$ is not an easy task. In this section, we determine diameter of SCCG generated by a few special seed graphs. 

We use the notation $G^{(0)}_{k^{i}}$ to denote the $k^{th}$ seed graph $G^{(0)}$ (having $n$ nodes) added in the $i$-th $(i\leq m)$ step in the formation of the corona graph $G^{(m)}$. The $j$-th node in $G^{(0)}_{k^{i}}$ is denoted  by $v^{(i)}_{k^{j}}$ where $j=1:n.$ Then we have the following lemma.

\begin{lemma}\label{Lem:LSD}
Let $V^{(i)}$ be the set of nodes added in SCCG $G^{(m)}$ in the $i$-th $(0\leq i\leq m)$ step of its formation.  Then the longest shortest distance between a pair of nodes in $\bigcup_{a=0}^{m}V^{(a)}-V^{(0)}$ is $2$.
\end{lemma}
\begin{proof}
Let the hub selected in $SC^{(1)}$ be $h^{(0)}_1$. We prove the lemma by determining the shortest distance between a node in $G^{(0)}_{k^{i}}$ ($k$ is fixed but $i\in[1,m]$) and any other node in $G^{(0)}_{l^{j}}$ ($j\in[1,m]$). Now we consider the following cases. 
\begin{itemize}
\item[(a)] $k\neq l:$ since each node of $V^{(i)}$ is connected to $h^{(0)}_1$, any node of $G^{(0)}_{k^{i}}$ and $G^{(0)}_{l^{i}}$ is connected to $h^{(0)}_1$. Hence, the shortest distance between the nodes of $G^{(0)}_{k^{i}}$ and $G^{(0)}_{l^{i}}$ is of length $2$.
\item[(b)] $k=l$ and $G^{(0)}=\{K_n,S_n\}:$ it is evident the shortest distance between nodes of $G^{(0)}_{k^{i}}$ is $1$ since the hub $h_1^0$ is connected to all the nodes.
\item[(c)] $k=l$ and $G^{(0)}$ be any seed graph: let $v^{(i)}_{k^{j}}$ and $v^{(i)}_{k^{j^{'}}}$ be any two nodes of $G^{(0)}_{k^{i}}$. Now, there are two sub-cases as follows.
\begin{itemize}
\item[(i)] if $(v^{(i)}_{k^{j}},v^{(i)}_{k^{j^{'}}})=1$ then the shortest distance is evidently $1$.
\item[(ii)]if $(v^{(i)}_{k^{j}},v^{(i)}_{k^{j^{'}}})=0,$ since $(v^{(i)}_{k^{j}},h^{(0)}_1)=1$ and $(v^{(i)}_{k^{j^{'}}},h^{(0)}_1)=1$, the shortest distance between the nodes is $2$ only.
\end{itemize}
\end{itemize}
Hence, the desired result follows.
\end{proof}

\begin{theorem}\label{Thm:KnSn}
Let $G^{(0)}$ be a graph such that degree of at least a single node of the graph is $(n-1)$. Then, the diameter of SCCG generated by $G^{(0)}$ is $2$.
\end{theorem}
\begin{proof}
Let the hub selected in $S^{(1)}$ be $h^{(0)}_1$. Since each node in the $SC^{(i)}$ is connected to $h^{(0)}_1$, every node is reachable from every other node in maximum $2$ steps for all $S^{(i)}$ after $G^{(i)}$ where $i\in[1,m]$. \end{proof}

Since $K_n$ and $S_n$ are also the graphs having the above property as stated in Theorem ~\ref{Thm:KnSn}, the diameter of any SCCG whose seed graph is $K_n$ or $S_n$ is also $2.$

\begin{figure*}
\subfloat[\label{9a}]{
  \includegraphics[width=0.1\textwidth]{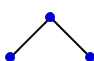}
}
\hspace{10mm}
\subfloat[\label{9b}]{
  \includegraphics[width=0.2\textwidth]{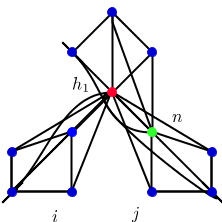}
}
\hspace{5mm}
\subfloat[\label{9c}]{
  \includegraphics[width=0.5\textwidth]{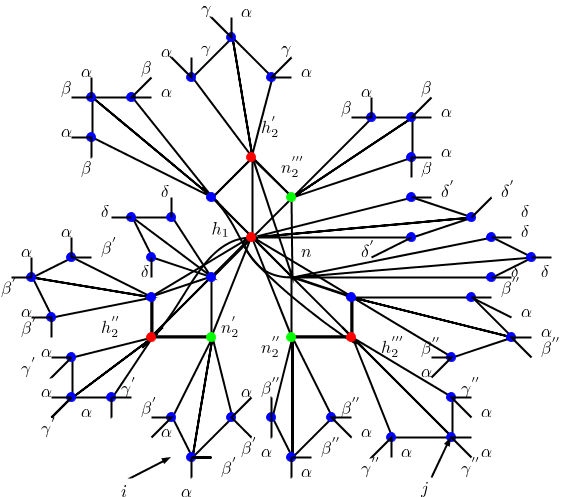}
}
\captionsetup{justification=raggedright,singlelinecheck=false}
\caption{\label{9} Example of diameter of SCCG (a) $S_3$ having its diameter as $2$.(b) In $SC^{(1)}$, the diameter is $2$ (here $h_1$ and $n$ are the hub and non-hub nodes respectively).  (c) In $SC^{(2)}$, the diameter again is $2$. Here, each node has either two or one stub attached with them. Nodes attached to skeletal nodes have two stubs which are shown as $\{\alpha,\beta\}$, $\{\alpha,\gamma\}$, $\{\alpha,\beta^{'}\}$, $\{\alpha,\gamma^{'}\}$, $\{\alpha,\beta^{''}\}$, $\{\alpha,\gamma^{''}\}$ while nodes attached to offshoots are shown as $\delta,\delta^{'},\delta^{''}$.The stubs ($\{\beta,\gamma\}$), ($\{\beta^{''},\gamma^{''}\}$) and ($\{\beta^{'},\gamma^{'}\}$) are connected to their pair of hubs and non-hubs as $(h_2^{'},n_2^{'''}),(h_2^{'''},n_2^{''})$ and $(h_2^{''},n_2^{'})$ respectively (edge from stubs to the nodes are not shown here to avoid visual complexity). The $\alpha$ stub is connected to hub $h_1$. The stubs $\delta,\delta^{'}$ and $\delta^{''}$ are connected to the pair of hub and non-hub as $(h_1,n)$. The red, green and blue colour nodes in the figure are hubs, non-hubs and simple nodes respectively.}
\end{figure*}
In Fig.~\ref{9}, we validate the above result for $G^{(0)}= S_3.$ We now derive the diameter of a SCCG when $G^{(0)}=\{_n,C_n\}$. Note that, we need to consider the shortest distance between any nodes of the node sets $V^{(0)}$ and $\bigcup_{a=0}^{m}V^{(a)}-V^{(0)}$ in order to determine the formula for the diameter.

\begin{theorem}\label{Thm:Cn}
The diameter of SCCG generated by $G^{(0)}=C_n$ is given as follows.
\begin{itemize}
\item[(a)] if $n=4$: the diameter is $2$ after $SC^{(1)},$ and $3$ after $SC^{(i)}$ when $i\geq 2$.
\item[(b)] if $n\in [5,7]:$ the diameter is $3$.
\item[(c)] if $n\ge 8:$ the diameter is $4$.
\end{itemize}
\end{theorem}
\begin{proof}
Let the hub selected in $SC^{(1)}$ be $h^{(0)}_1$. 
\begin{itemize}
\item[(a)] For $C_4:$ the longest shortest distances for each of the cases discussed above are as follows
\begin{itemize}
\item[(1)]Any pair of nodes in $V^{(0)}:$ Let $v^{(0)}_{0^{j}}\in V^{(0)}$ and $v^{(0)}_{0^{k}}\in V^{(0)}$ be two diametrically opposite nodes of $V^{(0)}$ such that $(v^{(0)}_{0^{j}},v^{(0)}_{0^{k}})=0$. Then the shortest distance is $2$ in between these nodes.
\item[(2)]Any pair of nodes in $\bigcup_{a=0}^{m}V^{(a)}-V^{(0)}:$ The proof is similar to the proof of Lemma \ref{Lem:LSD} and hence, the shortest distance is $2$.
\item[(3)] Between a node in $\bigcup_{a=0}^{m}V^{(a)}-V^{(0)}$ and a node in $V^{(0)}:$ Let $v^{(1)}_{k^{j}}\in V^{(1)}$ be a fixed node. Consider two cases as follows. Case I: $a=1$ and Case II: $a\ge 2$. 

We prove Case I as follows. Let $v^{(0)}_{0^{j}}\in V^{(0)}$ and $v^{(0)}_{0^{k}}\in V^{(0)}$ be two nodes of $V^{(0)}$ that are not hubs. Let $(v^{(1)}_{k^{j}},v^{(0)}_{0^{j}})=1.$ % i.e. the $k^{th}$ seed graph of $G^{(1)}$ is connected with $v^{(0)}_{0^{j}}$. 
We now want to find out the shortest distance between  $v^{(1)}_{k^{j}}$ and $v^{(0)}_{0^{k}}$. If $(v^{(0)}_{0^{j}},v^{(0)}_{0^{k}})=1$, then the shortest distance is $2$. If $(v^{(0)}_{0^{j}},v^{(0)}_{0^{k}})=0$, then both of the assumed nodes of $V^{(0)}$ are linked with $h^{(0)}_1$ and since $(h^{(0)}_1,v^{(1)}_{k^{j}})=1$, the shortest distance between the desired nodes is $2$. A similar reasoning could be used to prove the shortest distance as $2$ when one of the above assumed nodes of $V^{(0)}$ is a hub and the other node is either adjacent to it or not adjacent.

We now prove Case II i.e. $a\ge 2$. Let $v^{(2)}_{k^{j}}\in V^{(2)}$, $v^{(0)}_{0^{j}}\in V^{(0)}$, $v^{(1)}_{l^{j}}\in V^{(1)}$ such that $(h^{(0)}_1,v^{(0)}_{0^{j}})=0$ and $(v^{(1)}_{l^{j}},v^{(0)}_{0^{j}})=1$. The shortest path would be traced from $v^{(2)}_{k^{j}}$ to $v^{(0)}_{0^{j}}$ as: $(v^{(2)}_{k^{j}},h^{(0)}_1,v^{(1)}_{l^{j}},v^{(0)}_{0^{j}})$. Hence, the longest shortest distance is $3$.
\end{itemize}

\item[(b)] For $n\in[5,7]$, the longest shortest distance for each of the cases discussed above are as follows.
\begin{itemize}
\item[(1)]A pair of nodes in $V^{(0)}:$ Let $v^{(0)}_{0^{j}}\in V^{(0)}$ and $v^{(0)}_{0^{k}}\in V^{(0)}$ be two diametrically opposite nodes of $V^{(0)}$. The shortest distance is $2,3$ and $3$ in between these nodes of $C_5,C_6$ and $C_7$ respectively. 
\item[(2)] A pair of nodes in $\bigcup_{a=0}^{m}V^{(a)}-V^{(0)}:$ The proof is similar to the proof of Lemma \ref{Lem:LSD} and hence, the shortest distance is $2$.
\item[(3)] Between a node in $\bigcup_{a=0}^{m}V^{(a)}-V^{(0)}$ and a node in $V^{(0)}:$
Suppose $v^{(a)}_{k^{j}}$,$v^{(a)}_{l^p}\in V^{(1)}$ and $v^{(0)}_{0^{c}}$, $v^{(0)}_{0^{b}}\in V^{(0)}$, such that $(v^{(a)}_{l^p},v^{(0)}_{0^{b}})=1$ and $(v^{(a)}_{k^j},v^{(0)}_{0^{b}})=0$. Now, we want to find out the shortest distance between $v^{(a)}_{k^{j}}$ and $v^{(0)}_{0^{b}}$. If $v^{(0)}_{0^{b}}$ is a hub, shortest distance is $1$, if it is adjacent to a hub, the distance is $2.$ Otherwise, the shortest distance would be traced as: $(v^{(a)}_{k^{j}},h^{(0)}_1,v^{(a)}_{l^{p}},v^{(0)}_{0^{b}})$. Hence, shortest distance is $3$.

\end{itemize}
Hence, the diameter of SCCG where $G^{(0)}=C_n$ for $n\in[5,7]$ is $3$.
\item[(c)] For $n\ge 8:$ The longest shortest distance for each of the cases discussed above are as follows.
\begin{itemize}
\item[(1)] A pair of nodes in $V^{(0)}$: Suppose $v^{(a)}_{k^{j}}$,$v^{(a)}_{l^p}\in V^{(1)}$, $v^{(0)}_{0^{c}}$, $v^{(0)}_{0^{b}}\in V^{(0)}$, such that $(v^{(a)}_{k^j},v^{(0)}_{0^{c}})=1$, $(v^{(a)}_{l^p},v^{(0)}_{0^{b}})=1$, $v^{(0)}_{0^{c}}$ and $v^{(0)}_{0^{b}}$ are two diametrically opposite nodes of $V^{(0)}$. Now, we want to find out the shortest distance between $v^{(0)}_{0^{c}}$ and $v^{(0)}_{0^{b}}$. If $v^{(0)}_{0^{b}}$ is hub, shortest distance is $1$, if it is adjacent to hub, the distance is $2.$ Otherwise, the shortest distance would be traced as: $(v^{(0)}_{0^{c}},v^{(a)}_{k^{j}},h^{(0)}_1,v^{(a)}_{l^{p}},v^{(0)}_{0^{j}})$. Hence, the shortest distance is $4$.
\item[(2)]A pair of nodes in $\bigcup_{a=0}^{m}V^{(a)}-V^{(0)}:$ The proof is similar to the proof of Lemma \ref{Lem:LSD} and hence, the shortest distance is $2$.
\item[(3)] Between a node in $\bigcup_{a=0}^{m}V^{(a)}-V^{(0)}$ and a node in $V^{(0)}:$ The proof of diameter is similar as above of part (3) of (b).
\end{itemize}
Hence, the desired result follows.
\end{itemize}
\end{proof}

\begin{theorem}\label{Thm:n}
The diameter for all SCCG generated by the seed graph $G^{(0)}=_n$ are as follows.
\begin{itemize}
\item[(a)] the diameter is $3$ if $n=4$.
\item[(b)] for $n\in [5,7],$ the diameter is $3$ if the hubs which are selected at each step of the SC are adjacent to lowest degree nodes, otherwise the diameter is $4$.
\item[(c)] for $n\ge 8$ the diameter is $4$.
\end{itemize}
\end{theorem}
\begin{proof}
The proof is similar to the proof of Theorem \ref{Thm:Cn}.
\end{proof}

\begin{remark}
The SCCG generated by seed graphs as mentioned in of Theorem \ref{Thm:KnSn}, Theorem \ref{Thm:Cn} and Theorem \ref{Thm:n} exhibit constant diameters (except in $C_4$ where the constant diameter is observed after first step of the corona graph).
\end{remark}
\begin{figure}
\centering
\subfloat[\label{10a}]{
  \includegraphics[width=0.07\textwidth]{sample_graphs_for_diameter1.png}
}
\hspace{5mm}
\subfloat[\label{10b}]{
  \includegraphics[width=0.07\textwidth]{sample_graphs_for_diameter2.png}
}
\hspace{5mm}
\subfloat[\label{10c}]{
  \includegraphics[width=0.05\textwidth]{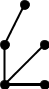}
}
\hspace{5mm}
\subfloat[\label{10d}]{
  \includegraphics[width=0.05\textwidth]{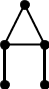}
}
\hspace{5mm}
\subfloat[\label{10e}]{
  \includegraphics[width=0.05\textwidth]{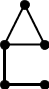}
}
\hspace{5mm}
\subfloat[\label{10f}]{
  \includegraphics[width=0.05\textwidth]{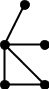}
}
\hspace{5mm}
\subfloat[\label{10g}]{
  \includegraphics[width=0.05\textwidth]{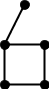}
}
\captionsetup{justification=raggedright,singlelinecheck=false}
\caption{\label{10}Examples of seed graphs having constant diameter (other than of cliques, trees, paths and circuits). Seed graphs in ~\ref{10c}, ~\ref{10d}, ~\ref{10e}, ~\ref{10f} and ~\ref{10g} are isomorphic forms of $G_{10},G_{12},G_{13},G_{14}$ and $G_{16}$ respectively as given in \cite{prvzulj2004modeling,prvzulj2007biological}.}
\end{figure}

We illustrate this remark with an example as shown in Fig.~\ref{9}. However, it is difficult to verify the constant diameter phenomena comprehensively for any arbitrary connected seed graph (for all $n$) empirically and analytically. We have verified this property empirically for a few other seed graph (for instance, the seed graphs of Fig.~\ref{10}). We further verified empirically that constant diameters are $2,2,3,3,3,2$ and $3$ for the seed graphs shown in Fig.~\ref{10a}, Fig.~\ref{10b}, Fig.~\ref{10c}, Fig.~\ref{10d}, Fig.~\ref{10e}, Fig.~\ref{10f} and Fig.~\ref{10g} respectively. Note that, the seed graphs (in Fig.~\ref{10}) i.e. Fig.~\ref{10c}, Fig.~\ref{10d}, Fig.~\ref{10e}, Fig.~\ref{10f} and Fig.~\ref{10g} are isomorphic to $G_{10},G_{12},G_{13},G_{14}$ (although diameter of $G_{14}$ could be proved analytically by Theorem \ref{Thm:KnSn}) and $G_{16}$ in \cite{prvzulj2004modeling,prvzulj2007biological}. It is also evident that the seed graphs given in Theorem \ref{Thm:KnSn}, Theorem \ref{Thm:Cn} and Theorem \ref{Thm:n} include the few Graphlets as given in \cite{prvzulj2004modeling, prvzulj2007biological} except $G_{10},$ $G_{12}$, $G_{13}$, $G_{16}$, $G_{19}$, $G_{20}$, $G_{21},$ and $G_{25}$. Here, one interesting point is to be reiterated that the SCCG model shows the property of constant diameter which is also a property of deterministic Kronecker graphs \cite{leskovec2010kronecker} and a few real world networks (with shrinking/constant diameters) \cite{leskovec2007graph}.
\begin{figure}
\centering
     \includegraphics[width=0.35\textwidth]{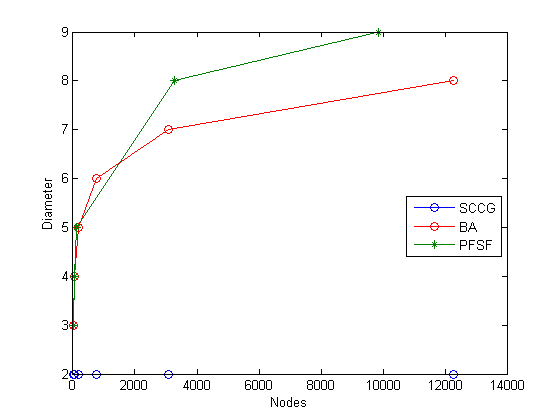} 
\captionsetup{justification=raggedright,singlelinecheck=false}
\caption{\label{11} Comparison of diameter of the SCCG with BA model and seudofractal scale-free (FSF).}
\end{figure}

We compared the diameter of SCCG derived by $G^{(0)}=K_3$ with that of networks generated by the BA and FSF models (as in Fig.~\ref{11}). However, it is evident that the diameter of the BA and FSF models are always increasing like corona graphs while for Kronecker graphs generated with a variation of $S_3$ having self-loops, the diameter is $2.$ 

In the following sub-sections, we compare the clustering coefficient, number of triangles and degree-degree correlation of SCCG generated by a triangle with that of networks generated by BA and FSF models.
\begin{figure*}
\centering
\subfloat[]   {%
   \includegraphics[width=0.32\textwidth]{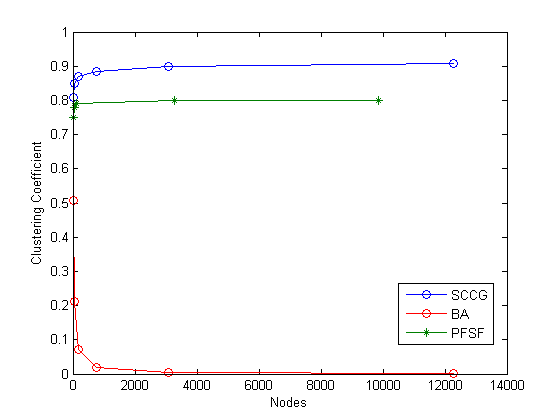} 
   \label{12a}
}
\hspace{0mm} 
\subfloat[]    {%
     \includegraphics[width=0.32\textwidth]{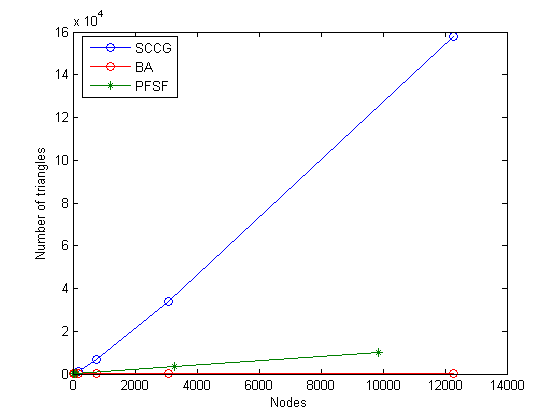} 
        \label{12b}
    }
\captionsetup{justification=raggedright,singlelinecheck=false}  
\caption{\label{12} Comparison of the (a) Clustering coefficient (b) frequencies of triangles for SCCG, BA and FSF models. Both the results are calculated with the help of Gephi \cite{bastian2009gephi}.}
\end{figure*}

\subsection{Clustering Coefficient and Frequency of triangles}
The clustering coefficient($C_v$) of a node $v$ in a network is defined as the probability that the neighbours of $v$ are themselves connected to each other. The average clustering coefficient($\overline{C}$) of a network is the average of the clustering coefficient of all nodes of the network \cite{albert2002statistical, newman2003structure}. In \cite{newman2003social}, the authors noted that social networks have high clustering in comparison to other networks, for instance, in the networks of movie actors \cite{watts1998collective}, company directors \cite{davis2001small}, co-authorship networks \cite{barabasi2002evolution} etc.

We numerically compute the clustering coefficient of SCCG determined by $G^{(0)}=K_3$ as the seed graph and compare it with that of other models (shown in Fig.~\ref{12a}). Since clustering coefficient reflects the frequency of the triangles in a network which is studied extensively in social network research \cite{durak2012degree}, we also compare the frequency of triangles in networks generated by BA and FSF models in Fig.~\ref{12b}. It shows from the results that the frequency of triangles in a SCCG is higher compared to other models. 
\subsection{Degree-degree correlation}
The degree-degree correlation of nodes in a network defines the notion of assortativity and disassortativity \cite{newman2002assortative} of a network. Assortativity signifies the existence of links between nodes of similar degrees whereas the existence of links between nodes of dissimilar degrees is known as disassortativity. Thus, disassortativity implies the propensity of connection of nodes of different types in a network. The correlation coefficient or assortativity coefficient ($r$) \cite{newman2010networks} is evaluated as
\begin{equation}
r=\frac{\sum_{ij}(A(G)_{ij}-\frac{k_i k_j}{2m}k_i k_j)}{\sum_{ij}(k_i\delta_{ij}-\frac{k_i k_j}{2m}k_i k_j)}
\end{equation}
where $A(G)_{ij}$ represents the $ij^{th}$ entry of adjacency matrix of the network, $k_i$ is the degree of node $i$ in the network, $m$ is the number of edges in the network, and $\delta_{ij}$ is the Kronecker delta function. If $r$ is positive, the network follows assortative mixing pattern of nodes whereas negative value of $r$ implies disassortative mixing of nodes in the network.
\begin{figure*}
\centering
\subfloat[]    {%
     \includegraphics[width=0.32\textwidth]{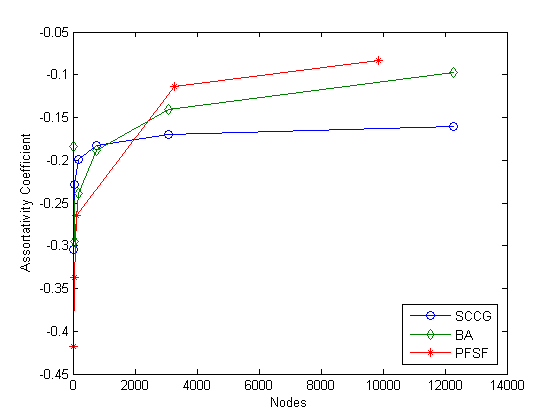} 
        \label{13a}
    }
\hspace{0mm}
\subfloat[]    {%
     \includegraphics[width=0.32\textwidth]{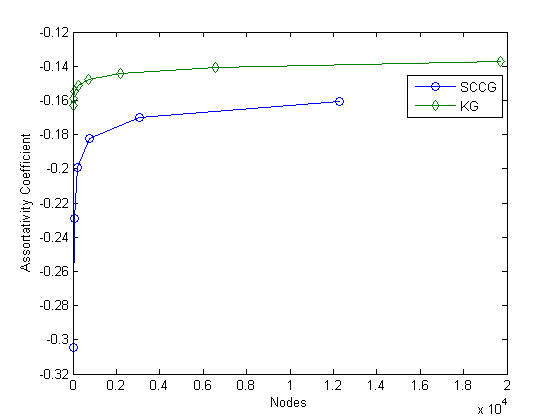} 
        \label{13b}
    }    
    
\captionsetup{justification=raggedright,singlelinecheck=false}  
\caption{\label{13} Comparison of the degree-degree correlation for (a) SCCG, BA and FSF models with $G^{(0)}=K_3$ as the seed graph. (b) SCCG and KG model where SCCG is generated with $G^{(0)}=S_3$ while $KG$ is generated with a seed graph $S_3$ having self-loop at each node.}
\end{figure*}
As we calculate the assortativity coefficient numerically for SCCG with $K_3$ as the seed graph and for networks generated by BA and FSF models, we observe that SCCG is showing disassortative mixing (Fig.~\ref{13a}) while it goes to zero for networks produced by BA and FSF models. However, the KG model shows disassortative mixing as given in Fig.~\ref{13b}. Note that, in Fig.~\ref{13b} the KG is derived by the initial graph $S_3$ with self-loop at each node, and SCCG is produced by the seed graph $S_3$. 

\section{\label{sec:con}Conclusion}
Recently, the phenomena of constant/shrinking diameter is observed in certain real networks. In this work, we mathematically modelled networks which inherit the property of power law degree distribution and constant diameter. We call these networks as Self-coordinated Corona Graphs (SCCG). Further, we numerically investigated the property of high clustering and disassortative mixing in these networks.

\end{document}